\documentclass[a4paper]{article}

\usepackage[english]{babel}
\usepackage[utf8x]{inputenc}
\usepackage{amsmath}%
\usepackage{amsfonts}%
\usepackage{amssymb}%
\usepackage{graphicx}
\usepackage[colorinlistoftodos, color=yellow]{todonotes}
\usepackage{nameref}
\usepackage[top=1in,bottom=1in, left=1in, right=1in]{geometry}
\usepackage{authblk}
\usepackage{libertine}

\usepackage{amsfonts}
\usepackage{amsthm}
\usepackage{amsmath}

\usepackage{thmtools, thm-restate}
\declaretheoremstyle[
qed=$\diamond$
]{mystyle}
\declaretheorem[numberwithin=section,qed=$\diamond$]{theorem}
\declaretheorem[unnumbered,name=Theorem,qed=$\diamond$]{theorem*}

\declaretheorem[unnumbered,name=Lemma,qed=$\diamond$]{lemma*}

\declaretheorem[unnumbered,name=Claim,qed=$\diamond$]{claim*}

\declaretheorem[unnumbered,name=Proposition,qed=$\diamond$]{proposition*}

\declaretheorem[unnumbered,name=Corollary,qed=$\diamond$]{corollary*}
\declaretheorem[numberwithin=section,qed=$\diamond$]{problem}
\declaretheorem[unnumbered,name=Problem,qed=$\diamond$]{problem*}
\declaretheorem[sharenumber=theorem,qed=$\diamond$]{fact}
\declaretheorem[unnumbered,name=Fact,qed=$\diamond$]{fact*}

\declaretheorem[unnumbered,name=Question,qed=$\diamond$]{question*}

\theoremstyle{definition}
\declaretheorem[numberwithin=section,qed=$\diamond$]{definition}
\declaretheorem[unnumbered,name=Definition,qed=$\diamond$]{definition*}

\newtheorem*{remark*}{Remark}

\newtheorem*{example*}{Example}

\usepackage{braket}
\usepackage{physics}
\usepackage{bbm}
\usepackage{hyperref}
\usepackage{enumerate} 
\usepackage[]{algorithm}
\usepackage{algorithmicx}
\usepackage[noend]{algpseudocode}

\usepackage{mathtools}

\renewcommand{\vec}[1]{\mathbf{#1}} 
\let \oldv \v
\renewcommand{\v}[1]{\vec{#1}} 

\newcommand{\E}{\mathop{\mathbb{E}}} 

\newcommand{\RR}{\mathbb{R}}
\newcommand{\R}[1]{\mathbb{R}^{#1}} 
\newcommand{\Rnonneg}{\R{}_{\geq 0}} 
\newcommand{\Rpos}{\R{}_{> 0}} 

\newcommand{\defeq}{\vcentcolon=} 

\newcommand{\diag}{\mathbf{diag}} 
\newcommand{\poly}{\mathrm{poly}} 

\newcommand{\inprod}[2]{\left\langle #1,#2 \right\rangle} 

\newcommand{\erdosrenyi}{Erd\H{o}s-R\'enyi\ } 

\newcommand{\half}{\frac{1}{2}} 

\newcommand{\union}{\cup} 
\newcommand{\djunion}{\sqcup} 

\newcommand{\comment}[1]{\ignorespaces} 
\newcommand{\Mapsto}{\rightarrow}

\newcommand{\effres}[1][{}]{\mathrm{R}^{#1}_\mathrm{eff}}

\newcommand{\Tail}{\mathsf{Tail}}

\title{How to Realize a Graph on Random Points}
\author{Saad Quader \thanks{Corresponding author. Email: \texttt{saad.quader@uconn.edu}. Web page: \href{https://saadquader.wordpress.com}{https://saadquader.wordpress.com} } }
\author{Alexander Russell \thanks{Email: \texttt{acr@cse.uconn.edu}}}
\affil{University of Connecticut}

\begin{document}
\maketitle

\begin{abstract}
We are given an integer $d$, a graph $G=(V,E)$, and a uniformly random embedding   $f : V \rightarrow \{0,1\}^d$ of the vertices.  We are interested in the probability that $G$ can be ``realized'' by a scaled Euclidean norm on $\RR^d$, in the sense that there exists a non-negative scaling $w \in \RR^d$  and a real threshold $\theta > 0$ so that
\[
 (u,v) \in E \qquad \text{if and only if} \qquad \Vert f(u) - f(v) \Vert_w^2 < \theta\,,
\]
where $\| x \|_w^2 = \sum_i w_i x_i^2$.  

These constraints are similar to those found in the Euclidean minimum spanning tree (EMST) realization problem. 
A crucial difference is that the realization map is (partially) determined by the random variable $f$. 

In this paper, we consider embeddings $f : V \rightarrow \{ x, y\}^d$ for arbitrary $x, y \in \RR$. 
We prove that arbitrary trees can be realized with high probability when $d = \Omega(n \log n)$. 
We prove an analogous result for graphs parametrized by the arboricity: 
specifically, we show that an arbitrary graph $G$ with arboricity $a$ can be realized with high probability when $d = \Omega(n a^2 \log n)$. 
Additionally, if $r$ is the minimum effective resistance of the edges, $G$ can be realized with high probability when $d=\Omega\left((n/r^2)\log n\right)$. 
Next, we show that it is necessary to have $d \geq \binom{n}{2}/6$ to realize random graphs, or $d \geq n/2$ to realize random spanning trees of the complete graph. 
This is true even if we permit an arbitrary embedding $f : V \rightarrow \{ x, y\}^d$ for any $x, y \in \RR$ or negative weights. 
Along the way, we prove a probabilistic analog of Radon's theorem for convex sets in $\{0,1\}^d$.

Our tree-realization result can complement existing results on statistical inference for gene expression data which involves realizing a tree, such as \cite{lp_tree}. 
\end{abstract}

\section{Introduction}\label{sec:intro}

A 2015 study considered the following problem involving gene expression data \cite{lp_tree}. 
We are given a rooted tree $T=(V,E)$ on $n$ objects (cell types); the tree arose from some underlying biological process among these objects. 
We are also given a one-to-one map $f : V \rightarrow \RR^d$, giving a data point (feature vector) associated with each objects. 
Let us write $f(V) \defeq \{f(u) \}_{u \in V}$. 
The authors expected that the data points should ``explain'' the tree in the following sense: 
that there should exist non-negative weights $w=\mqty(w_1, \cdots, w_d)$ such that the geometric minimum spanning tree (GMST) of the points $f(V)$ under the \emph{weighted Euclidean norm} $\Vert x \Vert_w$ will be identical to $T$, where $\Vert x \Vert_w \defeq (\sum_{i=1}^d{w_i x_i^2})^{1/2}$. 
If this is true, we say that ``$w$ \emph{realizes} $T$ on $f$ with dimension $d$.''

\paragraph{The EMST realization problem.}
The above problem immediately brings to mind the \emph{Euclidean minimum spanning tree (EMST) realization} problem, an important problem in graph drawing and VLSI circuit design \cite{realization-np-hard,king-realization-d3}. 
It says: 
\emph{Given a tree $T=(V,E)$ and an integer $d \geq 1$, can we find a one-to-one map $h : V \rightarrow \RR^d$ such that the unique GMST on the points $\{h(u)\}_{u \in V}$ under the $\ell_2^d$ norm is identical to $T$?} 

Let $\hat{w}$ be a scaling of $\RR^d$, i.e., it is a linear operator on $\RR^d$ whose matrix representation $\hat{W}$ is diagonal. 
Suppose the EMST realization algorithm outputs a map $h \defeq \hat{w}f$ for some $f : V \rightarrow \RR^d$. 
We can show (see Section~\ref{sec:link-emst-gene}) that finding an $h = \hat{w}f$ is equivalent to finding a weighted Euclidean norm $\Vert \vdot \Vert_w$ consistent with the map $f$. 
Here, we treat $w$ as a linear operator with a matrix representation $W \defeq \diag(w_1, \cdots, w_d)$, and define $W \defeq \hat{W}^2$ so that the desired weights $w_i$ are non-negative.
We show in Section~\ref{sec:link-emst-gene} that any solution to the two problems above must satisfy the following constraints:
\begin{align}
&\text{For all }u,v \in V\,, \qquad \Big( (u,v) \in E \qquad \text{if and only if} \qquad \Vert \hat{w}(f(u) - f(v)) \Vert_2^2 < \theta \Big) \qquad \text{for some }\theta > 0
\, .
\label{eq:emst}
\end{align}
There is a crucial difference between these two problems. 
In the EMST realization problem, we have to use the $\ell_2^d$ norm but are free to optimize $f$. 
In the problem in \cite{lp_tree}, however, the map $f$ is fixed but we are free to optimize a weighted $\ell_2^d$ norm.

\paragraph{The motivation behind this work.}
The \cite{lp_tree} paper uses a linear program to find a feasible set of weights while keeping the number of non-zero weights as small as possible. 
Without computing the weights, we cannot tell a priori whether a realization exists for a particular dimension $d$. 
One could argue that if we knew the distribution of the points $f(V)$, we could have used more appropriate techniques.
However, the problem description does not specify a process for generating the map $f$. 

Our work addresses this gap by defining and analyzing a problem where $f$ has a certain distribution.

\paragraph{Our problem: Graph realization with a random $f$.}
\emph{What can we say when the map $f : V \rightarrow \RR^d$ is uniformly random in some subset of $\RR^d$?} 
This is an intriguing question from a theoretical perspective. The random map $f$ would not depend on the edges of $T$. 
Moreover, the constraints (\ref{eq:emst}) do not mandate any structure on the edges $E$; they do not require $T$ to be a tree. 
This allows us to think about realizing an undirected graph $G = (V,E)$ on a set of random points. 
This is the central problem studied in this paper. 
Now, however, we require that the realization happens with high probability in the randomness in $f$.

\begin{problem}[Graph Realization on Random Points]
\label{prob:embed}
We are given an undirected graph $G=(V,E)$, a positive integer $d$, and a uniformly random embedding $f : V \rightarrow \{x,y\}^d$ for two arbitrary reals $x$ and $y$. 
We wish to find a weighted Euclidean metric $\norm{x}_w \defeq \sqrt{\sum{w_i x_i^2}}$ given by non-negative weights $w=\mqty(w_1, \cdots, w_d)$ which satisfies the constraints (\ref{eq:emst}) with probability
\[
R(G,d) \defeq \Pr_{f}[\text{ $G$ can be realized under $f$} ]
\,.
\]
We say that with probability $R(G,d)$, $G$ is \emph{realized} by $w$ on the embedding $f$ with dimension $d$. 
If $R(G,d) = 1-o(1)$, we simply omit the probability part.
\end{problem}

It will be instructive to think of $f$ as a random map from $V$ to $\{0,1\}^d$ or $\{\pm 1\}^d$. 
We are interested in a realization which is at least partly determined by a given random embedding. 
This aspect sets this problem apart: 
as far as we know, such a characterization has not been studied in the realization literature. 
In addition, the EMST realization problem concerns only trees, as does the problem in \cite{lp_tree}. 
There are several notions of realizing graphs, but none in the sense described above. 
Refer to Section~\ref{sec:lit-review} for the connections to some known problems.

There are some natural questions. 
Is there an algorithm to realize arbitrary graphs? 
What is the time complexity? 
How does that algorithm depend on the target dimension? 
Which role does the structure (e.g., largest degree, diameter, edge density, tree-width, etc.) play? 
While it is conceivable that a large $d$ would ``make things easy'' by allowing more degrees of freedom, it is not obvious ``how large'' a $d$ is necessary, or sufficient. 
Are there graphs that are ``hard to realize'' in the sense that they force \emph{every algorithm} to require a large $d$?
We address these questions in this work.

\subsection{Our Contributions}

We analyze two algorithms for Problem~\ref{prob:embed}, one for realizing trees and the other for realizing graphs. Both algorithms use only zero-one weights although we are allowed to use any nonnegative weights. The analysis reveals that using larger weights would not impact the bound on the dimension. It would, however, impact the threshold $\theta$ in (\ref{eq:emst}).

\begin{theorem}[Realizing a tree, see Theorem~\ref{thm:tree-second}]\label{thm:tree-second-informal}
There exists an algorithm that realizes any tree on $n$ vertices with high probability if the target dimension is $\Omega(n \log n)$.
\end{theorem}
\begin{theorem}[Realizing a graph, see Corollary~\ref{thm:graph}]\label{thm:graph-informal}
Let $G$ be an undirected graph on $n$ vertices, and let $a$ be the arboricity of $G$. There exists an algorithm that realizes $G$ with high probability if the target dimension is $\Omega(n a^2 \log n)$, which is at most $\Omega(n \vert E \vert \log n)$.
\end{theorem}
\begin{theorem}[Hard instances, see Corollary~\ref{thm:lowerbound} and Theorem~\ref{thm:lowerbound-tree}]\label{thm:lowerbound-informal}
It is impossible to realize an \erdosrenyi random graph on $n$ vertices if the target dimension is less than $\binom{n}{2}/6$.
It is impossible to realize a random spanning tree of a complete graph on $n$ vertices if the target dimension is less than $n/2$.
\end{theorem}

Theorem~\ref{thm:tree-second-informal} holds for weighted $\ell_1$ metric as well. The constant hidden under the $\Omega$ notation is $864$. 
The \emph{arboricity} of $G$ (Definition~\ref{def:arboricity}) can be interpreted as a measure of how sparse $G$ is. 
Since the arboricity of a tree is $1$, the bound on $d$ in Theorem~\ref{thm:graph-informal} implies the bound in Theorem~\ref{thm:tree-second-informal}. 
We also explore a connection between the bound on $d$ and the effective resistance of the edges of $G$. 
Theorem~\ref{thm:lowerbound-informal} complements Theorem~\ref{thm:graph-informal} by giving a lower/necessary bound on $d$. 
The statement holds even if $f:V \rightarrow \{x,y\}^d, x,y \in \RR$ is not random or negative weights are allowed.

\paragraph{An application.}
Theorem~\ref{thm:tree-second-informal} can serve as a sanity-check for experiments where such realizations are taken as an evidence that the graph and the points ``explain'' each other. 
For example, in \cite{lp_tree}, the authors asked whether a set of $n = 38$ data points living in $d = 22,215$ dimensions can explain a given tree $T$ on $n$ vertices. 
According to Theorem~\ref{thm:tree-second-informal}, we can realize an arbitrary $38$-vertex tree using $d^\prime$ dimensions on a random point-set where $d^\prime = 864 n\log n \approx 119,429$. 
Since $d \ll d^\prime$, one can argue that the realization---i.e., the inference---achieved in \cite{lp_tree} was ``not a fluke.'' 
Such an argument, however, is contingent on the tightest \emph{known} bound.


\subsection{Relating the \cite{lp_tree} Problem to EMST Realization}\label{sec:link-emst-gene} 
Let us write the EMST realization map $h = \hat{w} f$ where $f : V \rightarrow \RR^d$ is one-to-one and $\hat{w}$ is a non-negative scaling of $\RR^d$. The matrix representations of these maps are $H,\hat{W}$, and $F$, respectively, with $H = \hat{W}F$. Here, every vertex $u \in V$ is identified with a standard basis vector $\vec{u}$ of $\RR^n$. Moreover, $\hat{W}$ is a $d\times d$ diagonal matrix with non-negative entries. Suppose the constraints (\ref{eq:emst}) hold. 

First, we claim that $\Vert \hat{w}(f(u) - f(v)) \Vert_2 = \Vert f(u) - f(v) \Vert_w$ where $w = \mqty(w_1,\cdots, w_d)$. This can be seen by writing $W \defeq \hat{W}^2 = \diag(w_1, \cdots, w_d)$ and observing that 
\[
\Vert \hat{w}(f(u) - f(v)) \Vert_2^2 
= \Vert \hat{W}F(\vec{u} - \vec{v}) \Vert_2^2 
= (\vec{u} - \vec{v})F^T W F (\vec{u} - \vec{v}) 
= \sum_{i=1}^d{w_i (f(u)_i - f(v)_i)^2 } 
= \Vert f(u)-f(v) \Vert_w^2
\, .
\]

Thus the constraints (\ref{eq:emst}) imply that for every $(u,v) \in E$ and every $(u^\prime, v^\prime) \in \overline{E}$, the distance between $f(u)$ and $f(v)$---under the norm $\Vert \cdot \Vert_w$---is shorter than the distance between $f(u^\prime)$ and $f(v^\prime)$. 
It follows that the unique GMST on the points $h(V)$ under the $\Vert \vdot \Vert_2$ norm will be identical to $T$, as will be the unique GMST $T^\prime_w=(f(V), E^\prime)$ on the points $f(V)$ under the $\Vert \vdot \Vert_w$ norm. 
Here, by ``identical,'' we mean $(u,v) \in E$ if and only if $(f(u), f(v)) \in E^\prime$, and by ``unique,'' we mean $T^\prime_w$ will have the lowest total edge-length among all spanning trees on $f(V)$. 
Consequently, finding an EMST realization map $h$ is equivalent to finding a weighted Euclidean norm $\Vert \vdot \Vert_w$ consistent with the map $f$. 

\subsection{Realizing Trees}\label{sec:intro-tree}
We show in Section~\ref{sec:tree} that it suffices for the tree-realization algorithm, Algorithm~\ref{alg:tree-second}, if the entries of the random matrix $F$ come from any fixed two-element set $\{x,y\}$ where $x,y \in \RR$. However, to make the present exposition clear, let us assume that $F \in \{\pm 1\}^{d \times n}$ so that every vertex $u$ is mapped to a random point in $\{\pm 1\}^d$.

\paragraph{Lengths to inner products.} Recall that $\Vert \hat{W}F(\vec{u}-\vec{v}) \Vert_2^2 = \Vert \hat{W}F\vec{u}\Vert_2^2 + \Vert \hat{W}F\vec{v}\Vert_2^2 - 2 \langle F\vec{u} \vert \hat{W}^T\hat{W} \vert F\vec{v} \rangle$ where $\langle x, y \rangle$ denotes the usual inner product $x^Ty$ and $\langle x \vert A \vert y \rangle \defeq \langle x, A y \rangle$ for any matrix $A$. Since the length of every vector in $\{\pm 1\}^d$ is the same, the constraints (\ref{eq:emst}) is equivalent to saying that the weighted inner product between two vectors $F\vec{u}$ and $F\vec{v}$ using the weights $W = \hat{W}^T\hat{W}$ must be ``large'' if $(u,v) \in E$, and ``small'' if $(u,v) \in \overline{E}$. 

Now we can focus on the inner products instead of lengths. The entries in the random vector $F \vec{u}$ are independent and identically distributed Bernoulli random variables. Hence we can independently select a weight $w_i$ that is ``best'' for the coordinate $i$. The precise sense of ``best'' will be discussed in a moment. 

The inner product $\langle F \vec{u} \vert W \vert F \vec{v} \rangle$ is the sum of individual contributions $c_i$ from each coordinate $i$. Fix a coordinate $i$ and two vertices $u,v \in V$. The heart of the analysis is to show that in expectation, $c_i$ is ``large'' if and only if $(u,v)\in E$. Because each coordinate is independent, we can apply a Chernoff bound to show that the sum of these contributions -- i.e., the inner product $\langle F \vec{u} \vert W \vert F \vec{v} \rangle$ -- will be ``large'' if and only if $(u,v)\in E$.

\paragraph{Random sample strategy.} How do we select the weight $w_i$ for coordinate $i$? One way to do it is to pick $w_i$ so as to ``help'' a random tree-edge $(u,v)$ as follows: 1.) select an edge $(u,v) \in E$ uniformly at random, then 2.) set $w_i=1$ if $F\vec{u}_i = F\vec{v}_i$, and set $w_i=0$ otherwise. The rationale behind this ``random sample strategy'' is that this will make $c_i$ for this edge as large as possible (which is $1$). At the same time, it will randomize $c_i$ for all other vertex-pairs. Since every coordinate is pushing a tree-edge to the ``right direction,'' we hope that we can satisfy the constraints (\ref{eq:emst}) if we have a sufficiently large number of coordinates. Although the above idea works, the bound we get on $d$ is $\Omega(n^2 \log n)$ instead of the $\Omega(n\log n)$ bound which was promised by Theorem~\ref{thm:tree-second-informal}. (We omit the details.) 

\paragraph{Census strategy.} How do we improve the above strategy? Here is an idea: let us try to ``help'' multiple tree-edges at once. In particular, we set $w_i=1$ if a ``significant'' fraction of the tree-edges $(u,v)$ satisfy $F\vec{u}_i = F\vec{v}_i$; otherwise, we set $w_i=0$. This ``census strategy'' is detailed in Algorithm~\ref{alg:tree-second}, whose analysis leads to the desired bound of $d=\Omega(n\log n)$. This bound is only a $\log n$ factor away from the linear lower bound implied by Theorem~\ref{thm:lowerbound-tree}.

\subsection{Realizing Graphs}\label{sec:intro-graph}
When realizing a graph with cycles, the edges on a cycle are dependent in a non-trivial way. 
The census strategy ``touches'' multiple edges, and if two of them are on a cycle, a crucial argument breaks down in the proof of Claim~\ref{claim:inprod-tree-second-nonedge}. 
Not all hopes are lost, though, because the random sample strategy still works since it looks at only one edge at a time. However, it leads to a weak $\Omega(n^4\log n)$ bound on $d$. (Details omitted.)

We take the next natural choice: look at a family $\mathcal{A}$ of acyclic subgraphs of $G$ and invoke Algorithm~\ref{alg:tree-second} on a random member $A$ from this family. This works, and the bound we get depends on the probability that a given edge is contained in the sampled subgraph $A$. This is why $\mathcal{A}$ must cover every edge of $G$.

The best result comes when every edge belongs to exactly one member of $\mathcal{A}$. By necessity, $\mathcal{A}$ has to be a collection of edge-disjoint forests. This gives rise to the bound in Theorem~\ref{thm:graph-informal} containing the arboricity of $G$. If we take $\mathcal{A}$ to be the set of all spanning trees of $G$, the bound on $d$ is proportional to $1/r^2$ where $r$ is the smallest effective resistance among all the edges.

A simple tweak in Algorithm~\ref{alg:tree-second} allows us to realize the complement of any tree; this, in turn, allows us to realize any graph $G$ with $d=\Omega(na^2\log n)$ where $a=\min\{a(G), a(\overline{G})\}$ and $\overline{G}$ is the complement of $G$.

\subsection{Impossibility Results via a Geometric Interpretation}\label{sec:intro-impossibility}
The graph realization problem can be reduced to a hyperplane separation problem. Informally speaking, every constraint in (\ref{eq:emst}) specifies that a quantity of the form $\sum{w_i (F\vec{u} - F\vec{v})_i^2}$ be ``small'' if and only if $(u,v) \in E$. Observe that this quantity is the inner product of the vector $w=\mqty(w_1, \cdots, w_d)$ with the vector $g(u,v):=\big( (F\vec{u} - F\vec{v})_i^2 \big)_{i\in [d]}$. $G$ is realizable if there is a threshold $\theta$ and a vector $w$ satisfying $\langle w, g(u,v) \rangle > \theta$ if and only if $(u,v) \in E$.

The graph $G$ naturally colors the elements in $g(V \times V)$ as follows: an element is red if $(u,v) \in E$, and blue otherwise. 
The original EMST realization problem is equivalent to the following. 
First coloring the elements of $V\times V$ as red (edges) or blue (non-edges) according to $G$. 
Then we fix a hyperplane $h_w$ with its normal vector $w$ to the all-ones vector. 
Finally, we move the points around (via choosing an embedding $f$) so that points of different colors are on different sides of the hyperplane. 
In contrast, in Problem~\ref{prob:embed}, we first select ${n \choose 2}$ points from $\{0,s^2\}^d$ according to the random map $f$. 
Next, we color these in red or blue according to $G$. 
Finally, we find a hyperplane $h_w$ so that it perfectly separates the red point-set from the blue point-set.

Consider the two convex hulls pertaining to the red points and the blue points. 
If they intersect, no hyperplane could possibly realize $G$. 
Armed with this observation, we ask: \emph{If we use a random two-coloring, how likely is the event that a separating hyperplane would exist?} The examination in Section~\ref{sec:lowerbound} culminates in Corollary~\ref{thm:lowerbound} which states that the probability is $o(1)$. Consequently, an \erdosrenyi random graph (respectively, a random spanning tree) on $V$ cannot be realized if the target dimension $d$ is sub-quadratic (respectively, sub-linear) in $n$.

\subsection{A Probabilistic Analog of Radon's Theorem for $\{0,1\}^d$}
Radon's theorem (Theorem~\ref{thm:radon}) in convex geometry states that for every point-set $B$ of size $d+2$ in $\RR^d$, \emph{there exists} a red/blue coloring of the points so that the red convex hull intersects the blue convex hull. 
However, it does not give the probability that a \emph{random} red/blue coloring would result in the intersecting convex hulls. 
We ask the following: Suppose $B \subset \{0,1\}^d$ and that the red/blue coloring is uniformly random. 
How large does the set $B$ have to be so that with high probability, the two convex hulls intersect? 
We believe that this question---as well as the answer below---is interesting in its own right.

\begin{theorem}[Informal, see Theorem~\ref{thm:nonseparable}]
With high probability, a uniformly random red/blue coloring of a point-set containing at least $6d$ points in $\{0,1\}^d$ is not separable by any hyperplane in $\RR^d$ if $\vert B \vert \geq 6d$.
\end{theorem}

The proof of Theorem~\ref{thm:nonseparable} relies on counting the number of hyperplanes in $\RR^d$ that are ``pressed against'' exactly $d$ points in $\{0,1\}^d$. Every separating hyperplane implies a ``pressing'' hyperplane (Proposition~\ref{prop:sep-imply-press}). If there is no pressing hyperplane---which happens with high probability (Proposition~\ref{prop:no-pressing})---there can be no separating hyperplane.

\subsection{Related Problems}\label{sec:lit-review}
\paragraph{EMST realization.} 
Two factors play a key role  in determining whether an EMST is realizable: the largest degree $\Delta$ in the tree and the target dimensionality $d$. For $d=2$, solving the EMST realization problem is always possible if $\Delta \leq 5$, impossible if $\Delta \geq 7$, but the corresponding decision problem is NP-Hard if $\Delta = 6$ \cite{realization-np-hard}. The landscape for $d=3$ is also fragmented with results conditioned on the structure of the tree and the dimensions of the target space \cite{king-realization-d3}.

The EMST realization problem can be thought of as the inverse of the \emph{Euclidean Steiner Tree Problem}, which asks the following: given $n$ points in $\RR^d$, find a tree with the shortest total edge length.

\paragraph{Euclidean distance matrix realization.} 
Suppose we are given a matrix $\hat{D}$ containing the ``desired'' pairwise distances for a set of vertices $V$. 
To realize $\hat{D}$ in $\RR^d$, we have to map the vertices in $\RR^d$ such that the pairwise Euclidean distances among the mapped vertices equal the prescribed value in the distance matrix \cite{liberti2013realization}. 
In \cite{hendrickson-unique}, Hendrickson studied the conditions under which a graph has a unique realization in this sense. 
Although the EMST realization problem can be seen as a thresholded version of this distance matrix realization problem---the adjacency matrix of $T$ plays the role of the distance matrix $\hat{D}$---the adjacency matrix does not give a metric.
Hence the results concerning the distance matrix realization problem do not directly apply to the EMST realization problem.

\paragraph{Other areas.}
A \emph{structure preserving map} (SPE) of a graph $G$ into $\ell_2^d$ preserves some global topological structure of a set of high-dimensional data points $P$ while projecting them into a space of lower dimension \cite{SPE, NPE, GEE}. 
However, they infer the ``structure'' from $P$ itself whereas in our problem (Problem~\ref{prob:embed}), the structure $T$ is given and the data points $P$ are uniformly random.

Under a suitable formulation, the \emph{supervised metric learning problem} requires one to learn a weighted $\ell_2$ metric on a given point-set $P$ where the adjacencies $T$ are also given as an input \cite{metric_learning_paper}. However, this optimization problem is more similar to the situation in \cite{lp_tree} than to Problem~\ref{prob:embed} because the data points in a learning task are typically not random.



\subsection{Organization}
Section~\ref{sec:definitions} contains a precise definition of the graph realization problem. 
We analyze a tree-realizing algorithm in Section~\ref{sec:tree}. 
In Section~\ref{sec:graph}, we analyze an algorithm which realizes an arbitrary graph. 
The proof of the main impossibility result is outlined in Section~\ref{sec:lowerbound}. 
To make the exposition clear, some important proofs are pushed to the Appendix.

\section{Definitions}\label{sec:definitions}
We use $[d]$ to denote the set of first $d$ natural numbers, $\{1,\cdots, d\}$. $\inprod{x}{y}$ denotes the usual inner product between vectors $x$ and $y$. $d_T(u,v)$ is the length of the unique $u$-$v$ path in the unweighted tree $T$. We use the symbol $A\djunion B=C$ to denote a disjoint union of $A$ and $B$, or equivalently, a partition of $C$. $\Rpos$ denotes the positive reals, and $\Rnonneg$ denotes the nonnegative reals. For a matrix $A$, we write $A^{(i)}$ to denote the $i$th column of $A$. The expression $a\sim_U A$ denotes that the member $a$ is sampled uniformly at random from the set $A$. 

\begin{definition}[$d$-map, $d$-random map, and $(d,s)$-random map]\label{def:map}
Fix two arbitrary reals $x,y$. For any set $V$, let $f:V\rightarrow \{x,y\}^d$ be an map of $V$ into $\R{d}$ where each $f(u)$ is selected independently and uniformly in $\{x,y\}^d$. Then we call $f$ a \emph{$d$-random map}. If $f$ is not random, we call it a \emph{$d$-map} instead. We call $f$ a \emph{$(d,s)$-random map} if $\abs{x-y} = s$.
\end{definition}

\begin{definition}[Weighted Euclidean distance and its square]\label{def:dist}
Given a nonnegative vector $w=\mqty(w_1, \cdots, w_d)$, the \emph{weighted Euclidean norm} of $x \in \RR^d$ is defined as $\Vert x \Vert_w \defeq \left(\sum_i{w_i x^2 }\right)^{1/2}$. Given a $(d,s)$-map $f$, define the \emph{squared} Euclidean distance
\[
D(u,v) \defeq \Vert f(u) - f(u) \Vert_w^2 =  \sum_{i=1}^d{D_i(u,v)}
\, ,
\qquad
\text{where}
\qquad 
D_i(u,v) \defeq w_i(f(u)_i-f(v)_i)^2 \in \{0, s^2\}
\, .
\]
\end{definition}

\begin{definition}[Arboricity]\label{def:arboricity}
The \emph{arboricity} of an undirected graph $G=(V,E)$ is the minimum number of spanning forests needed to cover all the edges of the graph. Equivalently, it is the minimum number of forests $F_1, F_2, \cdots, F_a$ so that $E$ is the disjoint union of $F_1, F_2, \cdots, F_a$.
\end{definition}

\begin{definition}[Effective resistance]\label{def:effres}
Let $G=(V,E)$ be an undirected graph corresponding to an electrical network where each edge contains a unit resistance. For every vertex $u \in [n]$, let $\vec{u}$ be the $u$th standard basis vector of $\RR^n$ i.e., $\vec{u}_u=1, \vec{u}_v = 0$ for all $v \neq u$. Let $A$ be the adjacency matrix of $G$ and let $D$ be a diagonal \emph{degree matrix} of $G$ defined as $D(u,u) = \deg(u)$. Then the matrix $L=D-A$ is called the \emph{Laplacian matrix} of $G$. Let $L^+$ be the Moore-Penrose pseudoinverse of $L$. ($\rank(L)=n-c$ where $c$ is the number of connected components of $G$.) The \emph{effective resistance} between two vertices $u,v$ is given by \[\effres(u,v) = (\vec{u}-\vec{v})^T L^+ (\vec{u}-\vec{v})\, .\]
\end{definition}

The effective resistance is intimately linked with many combinatorial properties of a graph. See Ellens et al. \cite{ellens2011effective} for further reading. We use the following fact in this paper.

\begin{fact}\label{fact:effres}
Let $\mathcal{T}$ be the set of all spanning trees of the undirected graph $G=(V,E)$. Then
\[
\Pr_{T \sim_U \mathcal{T}} [e \in T] = \effres(e)\, .
\]
Moreover, $\effres(e) \geq 2/n$ for any $e \in E$.
\end{fact}

\section{Realizing a Tree}\label{sec:tree}

\newcommand{\Gap}{\mathrm{Gap}}
\begin{definition}[Gap and total gap]\label{def:deltai}
Let $e=(u,v)\in T$ and $e^\prime=(u^\prime,v^\prime)\not \in T$ be two arbitrary vertex pairs. The \emph{gap} between these two vertex pairs at coordinate $i$ is
\begin{align}\label{eq:gap_basic}
\delta_i(e,e^\prime) 
&\defeq \E_f\,[D_i(e^\prime) - D_i(e)]
\, ,
\end{align}
where $D(\vdot, \vdot)$ is defined in Definition~\ref{def:dist}. Define the \emph{total gap} between $e,e^\prime$ as
\begin{align}\label{eq:gap-total}
\Delta(e,e^\prime) 
&\defeq \sum_i{\delta_i(e,e^\prime)}
\, .
\end{align}

\end{definition}

Suppose we want to realize a tree $T$ using only Boolean weights. Only the coordinates with weight $1$ will contribute in the distance. We want to select the coordinates in such a way that the expected distance of an edge $e$ is pushed away from the expected distance of a non-edge $e^\prime$. This is the same as trying to enforce a large gap $\delta_i(e,e^\prime)$ at each coordinate which, by the linearity of expectation, would imply a large total gap $\Delta(e,e^\prime)$. This deterministic strategy is formalized in Algorithm~\ref{alg:tree-second} below.

\begin{algorithm}
\caption{RealizeTree$(T,f)$}
\label{alg:tree-second}
\begin{algorithmic}    
	\Require{$f$, a $(d,s)$-map from $V$ to $\RR^d$}
    \State Pick any real $\alpha \in (0, 1/2)$; in particular, $\alpha = 1/4$ works
    \For{$i\in [d]$ independently} 
        \State Let $p_i$ be the fraction of edges $(u,v)\in T$ such that $f(u)_i=f(v)_i$.
        \State Assign $w_i$ independently according to the following rule:
\[
w_i\defeq \left\{\begin{matrix}
1 &\text{ if } 1/2+\alpha/ \sqrt{n} \leq p_i \leq 3/4 \,,\\
0 &\text{ otherwise.}
\end{matrix}\right.
\]
   \EndFor
   \end{algorithmic}
\end{algorithm}

We devote the rest of this section analyzing Algorithm~\ref{alg:tree-second}.

\newcommand{\PrAgree}{\mathsf{PrAgree}}
\begin{definition}[Agreement probability]\label{def:pragree}
For any $e=(u,v) \in V$, define the \emph{agreement probability} as 
\[
\PrAgree(e,i)\defeq \Pr_f\,[f(u)_i=f(v)_i\mid w_i=1]
\,.
\]
\end{definition}

\begin{definition}[Weight selection probability, $q$]\label{def:q}
For Algorithm~\ref{alg:tree-second}, define the \emph{weight selection probability} 
\[
q \defeq \Pr_f[w_i=1]
\,.
\]
\end{definition}

When $e \in E,e^\prime \in \overline{E}$ are identified, we can expand Equation~\ref{eq:gap_basic} to show that
\begin{align}\label{eq:gap}
\delta_i(e,e^\prime) 
&= s^2q\, (\PrAgree(e,i)-\PrAgree(e^\prime,i))
\,.
\end{align}

A bad event occurs when there exist two vertex pairs $e\in E, e^\prime \in \overline{E}$ with $D(e^\prime) \leq D(e)$. 
Our argument for proving Theorem~\ref{thm:tree-second} has two steps. 
In the first step, we prove that for any fixed vertex pairs a bad event does not occur in expectation. 
This is equivalent to showing that the total gap $\Delta(e,e^\prime)$ is large. 
The second step has two phases. 
First, we bound the ``bad probability'' for a given vertex-pair $(u,v) \in V\times V$ via a Chernoff bound. 
Finally, we bound the total bad probability by applying a union bound over all vertex-pairs. 
Requiring that this probability be $1 - 1/n$ gives a bound on $d$.

\subsection{Step One: Proving that the Total Gap is Large}
Fix two vertex pairs $e\in E$ and $e^\prime \in \overline{E}$. 
The quantity $Z = \sum_i{D_i(e^\prime) - D_i(e)} = D(e^\prime) - D(e)$ is the sum of $d$ independent (but not identically distributed) Bernoulli random variables since $\{w_i\}$ are independent. 
We proceed by showing that the expectation of the $i$th component of this sum---i.e., $\delta_i$---is ``large.'' 
This implies that $D(e^\prime)$ is larger than $D(e)$ in expectation. 
Next, a Chernoff bound on $Z$ would reveal that $Z$ is unlikely to be ``too small'' compared to its expectation $\E Z = \Delta(e,e^\prime)$. 
Equivalently, with ``large'' probability, the length of the edge $e$ will be strictly shorter than the length of the non-edge $e^\prime$. 
This satisfies the constraints on the lengths of $e,e^\prime$ imposed by (\ref{eq:emst}).

Suppose Algorithm~\ref{alg:tree-second} assigns $w_i=1$. We want a lower bound on the gap $\delta_i \defeq \delta_i(e,e^\prime)$, or more appropriately, on the quantity $\PrAgree(e,i)-\PrAgree(e^\prime,i)$. Since $w_i=1$, we have seen exactly $p_i\abs{E}$ edges of $T$ to have the same values at both endpoints. For any two vertices $a,b\in V$, how does $\PrAgree( (a,b),i)$ depend on $p_i$? The answer is given by the following claim.

\begin{restatable}[Decaying correlation]{claim}{thmNonEdge}
\label{claim:inprod-tree-second-nonedge}
Let $T=(V,E)$ be a tree, and $p_i\in (1/2,1]$ be some positive real. Fix a coordinate $i$. Let $f_i:V\Mapsto \{x,y\}$ be a random variable defined as $\Pr[f(u)_i=x]=\Pr[f(u)_i=y]=1/2$ where $x,y\in \R{}$. Suppose, in an instance of $f(V)_i$, there are exactly $p_i\abs{E}$ edges having the same values at both endpoints. Let $(u,v)\in V\times V$ be an arbitrary vertex pair. Then, \[
\PrAgree(u,v,i)
=\half \left(1+(2p_i-1)^t\right)\] where $t$ is the length of the unique $u$-$v$ path along $T$.
\end{restatable}
We remark that the proof of the above claim is the only portion of our analysis which requires $T$ to be a tree. Claim~\ref{claim:inprod-tree-second-nonedge} implies that
\[
\PrAgree(e,i)-\PrAgree(e^\prime,i)
= \frac{(2p_i-1)-(2p_i-1)^t}{2} 
\geq \frac{(2p_i-1)-(2p_i-1)^2}{2}
= (2p_i-1)(1-p_i)
\, ,
\]
since $t\geq 2$ for $e^\prime\not \in T$ and $(2p_i-1)\leq 1$. It follows that 
\begin{align}\label{eq:deltai-simplified}
\delta_i &\geq s^2q(2p_i-1)(1-p_i)
\,.
\end{align}
However, we want an expression for the right hand side which does not depend on $i$ so that the sum $\sum{\delta_i}$, in turn, does not depend on $i$. After some calculations we get the following result; we defer the proof till Section~\ref{sec:tree-proofs}.

\begin{restatable}[Bounds on $\delta_i$, $p_i$, and $q$]{claim}{thmDeltai}
\label{claim:deltai}
The probability $p_i$ in Algorithm~\ref{alg:tree-second} is less than $1-\alpha/\sqrt{n}$. 
Moreover, the probability $q:=\Pr[w_i=1]$ is at least $1/2-\alpha - 2^{(H(1/4)-1)n}$ and at most $1-p=1/2-\alpha/\sqrt{n}$. 
Here, $H(\vdot)$ is the binary entropy function. 
In particular, $q\geq 1/6$ when $n\geq 20$ and $\alpha=1/4$. 
The gap $\delta_i = \Omega(1/\sqrt{n})$ when $\alpha, q, s$ are constants. 
Specifically, $\delta_i \geq s^2q(2\alpha/\sqrt{n})(1-\alpha/\sqrt{n})$. 
\end{restatable}

\subsection{Step Two: Bounding the Bad Probability via Chernoff/Union Bound}
We have already seen that for two fixed vertex pairs $e \in E$ and $e^\prime \in \overline{E}$, the gap between their respective expectations, i.e., $\Delta(e,e^\prime)$, is large. 
Let $\theta \defeq \Delta(e,e^\prime)/2$ be the midpoint of this gap. 
A \emph{bad event} occurs when either $D(e) > \theta$ or $D(e^\prime) < \theta$. 
The probability of an individual bad event can be obtained via the Chernoff-Hoeffding bound. 
Note that there can be at most ${n \choose 2}$ bad events. 
The probability that no bad event occurs can be found via a union bound. 
By setting this probability to at most $1 - 1/n$, we get a bound on $d$. 
The exact statement is recorded the following lemma; we defer its proof till Appendix~\ref{sec:tree-proofs}.

\begin{restatable}[Bounding $d$ from gap $\delta_i$]{lemma}{thmEdgeChernoff}
\label{lemma:edge-chernoff}
Let $f$ be a $(d,s)$-random map of the vertices $V$. 
Let $w_1, \cdots, w_d$ be the weights from Algorithm~\ref{alg:tree-second} invoked on the tree $T=(V,E)$ and the embedding $f$. 
Define $\delta \defeq \inf{\delta_i(e,e^\prime)}$ 
where the infimum is taken over all $i \in [d], e\in E$, and $e^\prime \in \overline{E}$. 
If $d\geq (6s^4\log n)/\delta^2$, the constraints (\ref{eq:emst}) are satisfied
with probability $1-1/n$ over the random choice of $f$ 
with $\theta = d\delta/2$.
\end{restatable}

\subsection{Main Theorem} 
\begin{restatable}[Realizing a tree]{theorem}{thmTree}
\label{thm:tree-second}
Suppose $d, n \in \mathbb{N}, n \geq 20, d=\Omega(n\log n)$. Let $T=(V,E)$ be a given tree on $n$ vertices. Let $f$ be a given $(d,s)$-random map of $V$. 
Then, $R\left( T, d \right) \geq 1-1/n$. 
In particular, Algorithm~\ref{alg:tree-second}, when using the parameter $\alpha = 1/4$, 
runs in time $nd$ and generates the weights $w_1, \cdots, w_d \in \{0, 1\}$ such that with probability $1-1/n$, 
the constraints (\ref{eq:emst}) are satisfied for some $\theta > 0$. 
The absolute constant hidden under the $\Omega$ notation in the expression of $d$ is $864 = 6/(2q\alpha)^2$ 
where $\alpha = 1/4, q = 1/6$ according to Claim~\ref{claim:deltai}; in particular, this constant is independent of the choice of $s$.
\end{restatable}

\begin{proof}
Let $n\geq 20$ and $\alpha=1/4$. By Claim~\ref{claim:deltai}, $q\geq 1/6$. 
Ignoring the $o(n)$ term in the expression of $\delta_i$ from Claim~\ref{claim:deltai}, we get 
$\delta_i
\gtrapprox 
\delta 
= s^2 q (2\alpha/\sqrt{n})
\geq s^2/12\sqrt{n}
$. 
The bound on $d$ from Lemma~\ref{lemma:edge-chernoff} gives
$d 
\geq 6s^4\log n / \delta^2
= 6s^4\log n ( 12\sqrt{n}/s^2 )^2
= C n\log n
$ where 
$
C = 6 \times 12^2 = 864
$.
This $d$ is sufficient so that the weights generated by Algorithm~\ref{alg:tree-second} with $\alpha = 1/4$ satisfy the constraints in Equation~(\ref{eq:emst}) with probability $1-1/n$. 

Using the expression of $\theta$ from Claim~\ref{lemma:edge-chernoff} from Claim~\ref{claim:deltai}, we get
\[
\theta
= \frac{d\delta}{2} 
= d s^2 q (2\alpha/\sqrt{n})(1-\alpha/\sqrt{n})/2
= \frac{d s^2}{4} q \left(\frac{1}{\sqrt{n}}-\frac{1}{4n} \right)
\, ,
\]
where 
\[
q \in \left( \frac{1}{4} - 2^{(H(1/4) - 1)n}, 1/2 - \frac{\alpha}{\sqrt{n}} \right)
\]
using Claim~\ref{claim:deltai}.

\end{proof}

\paragraph{Some remarks.} 
A weighted $\ell_1$ distance between two points $f(u), f(v)$ is defined as $\sum_{i}{w_i \abs{f(u)_i-f(v)_i}}$. 
It is not hard to see that if we use this metric in the preceding analysis, $s$ would appear as a linear factor in the expression of $\delta_i$ (from Definition~\ref{def:deltai}) since $f$ is a $(d,s)$-map. In addition, since the final bound on $d$ does not depend on $s$, an algorithm which realizes $T$ with a Boolean-weighted $\ell_2$ norm for a given $d$ would also work for a Boolean-weighted $\ell_1$ norm with the same $d$. However, the expression for the threshold $\theta$ would be affected since it depends on $s$. We omit further details.

If we modify Algorithm~\ref{alg:tree-second} to tally edge-disagreements instead of edge-agreements, we would realize the \emph{complement} of $T$. 
The factor $n=(\sqrt{n})^2$ in the bound $d=\Omega(n\log n)$ in Theorem~\ref{thm:tree-second} is an artifact of the algorithm used to realize $T$. 
In particular, it comes from the bias $p=1/2+O(1/\sqrt{n})$ in Algorithm~\ref{alg:tree-second}. 
The $\log n$ factor in the bound is an artifact of the $1-1/\poly(n)$ probability required from the Chernoff bound in the proof of Lemma~\ref{lemma:edge-chernoff}, and that there are $\poly(n)$ vertex-pairs in the union bound. 
It is hard to see how to improve the the current analysis without making a non-trivial change in Algorithm~\ref{alg:tree-second}. 

The bound on $d$ does not depend on $s$. 
Consequently, it would remain unchanged as long as $w_i \in \Rnonneg$ since such a scaling would simply scale $s$.


\section{Realizing a Graph}\label{sec:graph}
Let us elaborate on our discussion in Section~\ref{sec:intro-graph}. As in Section~\ref{sec:intro-tree}, suppose the set of random points are $F \in \{\pm 1\}^{d \times n}$. The analysis of the census strategy in the proof of Claim~\ref{claim:inprod-tree-second-nonedge} requires that the graph being realized is indeed a tree. Let us define the \emph{edge sign} $\sigma_i(u,v) \defeq +1$ if $u_i=v_i$, and $-1$ otherwise. The main observation in that proof is the following: For any $i\in [d]$, the uniform distribution of coordinate-values $F\vec{u}_i \in \{\pm 1\}, u \in V$ is identical to the uniform distribution of the edge signs $\sigma_i(e) \in \{\pm 1\}, e\in E$ coupled with a random assignment $F\vec{r}_i \in \{\pm 1\}$ to an arbitrary vertex $r\in V$. 

This observation, however, works only when $T$ is a tree; it breaks down if we want to realize a graph $G$ which contains a cycle. For example, suppose $G$ contains a triangle $(u,v,w)$. For every coordinate $i$, if $\sigma_i(u,v)=\sigma_i(v,w)$ then $u_i$ must equal $w_i$. In general, $\prod_{e\in C}{\sigma_i(e)}=1$ for every cycle $C$. Due to this correlation in coordinate values along a cycle, a uniform distribution of the coordinate values does not translate to a uniform distribution on the edge signs of $G$. Consequently, the census strategy is not applicable when $G$ contains a cycle.


The random sample strategy mentioned in Section~\ref{sec:intro-tree}, however, is immune to any correlation. It samples a single edge. By this virtue it is oblivious to any structure in the graph. We use this observation to devise an idea: \emph{what if we use an acyclic subgraph as a representative of $G$?} 
\paragraph{A strategy.}
Let $\mathcal{A}$ be a collection  of acyclic subgraphs of $G$. 
We would sample a member $A$ from $\mathcal{A}$ uniformly at random and run the tree-realization algorithm on $A$. 
This eliminates all cycles from our view, but it is not obvious that the resulting weights would satisfy the edges not on the subgraph. 
It turns out that the gap between the two kinds of inner products (edges vs. non-edges) depends on the probability that a given edge is included in the uniformly sampled member $A$. 
This is why $\mathcal{A}$ must cover every edge of $G$. 
This strategy is applied by the following algorithm. 


\begin{algorithm}
\caption{RealizeGraph$(G,f,\mathcal{A})$}
\label{alg:graph}
\begin{algorithmic}   
	\Require{$f$, a $(d,s)$-map from $V$ to $\RR^d$}
	\Require{$\mathcal{A}$, a family of acyclic subgraphs of $G$ such that every edge $e\in G$ belongs to at least one member of $\mathcal{A}$.}
	\State Sample an element $A$ uniformly at random from $\mathcal{A}$
    \State Invoke $\mathrm{RealizeTree}(A, f)$
\end{algorithmic}
 \end{algorithm}

The members of $\mathcal{A}$ do not have to be trees: they could be a single edge, a subtree, a forest, a matching, etc. In particular, $\mathcal{A}$ can contain multiple kinds of acyclic subgraphs as long as their union covers all edges. 

\begin{restatable}{claim}{claimGraph}
\label{claim:graph}
Algorithm~\ref{alg:graph} realizes $G$ with $\displaystyle d=\Omega(\frac{n}{r^2}\log n)\, , \text{ where} \quad r\defeq \min_{(u,v)\in E}{\Pr_{A \sim \mathcal{A}}[(u,v)\in A]}$.
\end{restatable}

\begin{restatable}[Realizing a graph]{corollary}{thmGraph}
\label{thm:graph}
For every graph $G=(V,E)$ on $n$ vertices, $n\geq 20$, $R(G,d)\geq 1-o(1)$ if $d=\Omega(n\abs{E}\log n)$. In particular, the weights generated by Algorithm~\ref{alg:graph} can realize $G$ with probability at least $1-1/n$ with $d=\Omega(n a^2\log n)$ where $a$ is the arboricity of $G$. 
\end{restatable}

\begin{proof}
Recall the definition of the arboricity (Definition~\ref{def:arboricity}). 
We can take $\mathcal{A}=\{\phi_i\}_{i=1}^a$ as the set of all edge-disjoint forests of $G$ so that $E=\union_i{\phi_i}$. The cardinality of $\mathcal{A}$ is the arboricity of $G$, and is denoted by $a\defeq a(G)$. The edge-disjointedness implies that every edge belongs to a unique forest $\phi_i$, and hence $r(e)=1/a$.  It follows that $d=\Omega(na^2\log n)$. In the worst case, $d=\Omega(n\abs{E}\log n)$ since $a\leq \left\lceil \sqrt{\abs{E}/2} \right\rceil$ using the bound in \cite{arboricity_bound}.
\end{proof}

It is easy to see that Theorem~\ref{thm:tree-second} is a special case of Corollary~\ref{thm:graph} because then $\mathcal{A}$ contains only one member, the tree $T$ itself.

\paragraph{A connection with effective resistance.} Suppose we take $\mathcal{A}$ as the set of all spanning trees of $G$. Using Fact~\ref{fact:effres}, we can see that $r(u,v)=\effres(u,v)$. This gives $d=\Omega\left( (n \log n)/\left(\min_e{\effres(e)}\right)^2 \right)$. Since $\effres(e) \geq 2/n$, $d = \Omega(n^3\log n)$ in the worst case. This bound is weaker than what we get if we use the arboricity in the proof of Corollary~\ref{thm:graph}.


\section{Realizing Random Graphs and Trees}\label{sec:lowerbound}

Let us make concrete the notion of ``linear separability'' which is at the center of our argument.
\begin{definition}[Linear Separability and Bipartition]\label{def:linsep}
Two point-sets $A, B\in \R{d}$ are \emph{linearly separable} (or \emph{separable} in short) if there exists a hyperplane with a normal vector $w$ such that $\inprod{a}{w} \leq \inprod{b}{w}$ for all $a\in A, b\in B$. A \emph{bipartition} of a point-set $B$ is a disjoint union of two convex subsets $B_1\djunion B_2=B$ where the subsets are separable.
\end{definition}

As we explained in Section~\ref{sec:intro-impossibility}, it is possible to cast the realization problem in Definition~\ref{prob:embed} as a question about separating two point-sets using a hyperplane. If $f$ is a $(d,s)$-map, the map $g$ from Section~\ref{sec:intro-impossibility} becomes 
\begin{align}
&g : V\times V \rightarrow \{0,s^2\}^d, \qquad g(u,v) = \left( w_i\left(f(u)_i-f(v)_i\right)^2 \right)_{i=1}^d \label{eq:g}
\, 
\end{align}
for every vertex pair $(u,v)\in E$ and $(u,u^\prime)\in \overline{E}$. 
Notice that the range of $g$ is $\{0,s^2\}^d$, which is the same as the Boolean hypercube scaled by $s^2$. 

If $G$ is a random \erdosrenyi graph, it would induce a random assignment on the points $g(V\times V)$ into convex sets $g(E)$ (imagine red) and $g(\overline{E})$ (imagine blue). 
Also note that the number of hyperplanes supported by $d$ points in the Boolean hypercube $C_d$ is bounded. 
This allows us to use a counting argument to show that with high probability in the randomness in $G$, the convex hulls of $g(E)$ and $g(\overline{E})$ will intersect if $d$ is ``small.''

The above argument does not depend on any structure in $G$ except that it is a random graph. 
Thus we can take $f$ to be arbitrary and allow the weights $w_i$ to be arbitrary reals. 

\begin{restatable}[Probabilistic version of Radon's Theorem]{theorem}{thmNonSeparable}
\label{thm:nonseparable}%
For any fixed $x,y \in \R{}$, let $B$ be an arbitrary subset of $\{x,y\}^{d}$. Create a uniformly random partition $E \djunion \overline{E}=B$ by independently setting $\Pr[b\in E]=\Pr[b\in\overline{E}]=1/2$ for every $b\in B$. If $\vert B \vert \geq 6d$, the convex hulls of $E$ and $\overline{E}$ intersect with probability at least $1-o(1)$.
\end{restatable}
A discussion and proof of Theorem~\ref{thm:nonseparable} is presented in Appendix~\ref{sec:lowerbound-proofs}. 

\begin{restatable}[Realizing a random graph]{corollary}{thmLowerbound}
\label{thm:lowerbound}%
Let $G\sim \mathcal{G}(n, 1/2)$ be an \erdosrenyi random graph with $n\geq 6$. Let $d\leq \binom{n}{2}/6$ be a positive integer. With probability at least $1-o(1)$ in the randomness of $G$, $G$ is not realizable under any $d$-map $f$ and any weights $w_i\in \R{}, i\in [d]$. This means $R(G,d)=0$.
\end{restatable}

\begin{proof}
Sample an \erdosrenyi random graph $G=(V,E)\sim \mathcal{G}(n,1/2)$ where $n\geq 6$. 
Also, let  $f:V\Mapsto \{x,y\}^{d}$ be an arbitrary map with $d\leq \binom{n}{2}/6$ and $x,y\in \R{}$. 
Since $E$ is a uniformly random subset of $V\times V$, we can invoke Theorem~\ref{thm:nonseparable} to show that with high probability, the random partition $E\djunion \overline{E}$ of the map $B=g(V\times V)$ is not linearly separable. 
Consequently, there exists no hyperplane (indicated by $w\in \R{d}$) that separates $E$ from $\overline{E}$. 
Recall that our definition of linear separability has inequality constraints. 
If these constraints cannot be satisfied by any hyperplane, it follows that the strict inequality constraints of Equation~(\ref{eq:emst}) cannot be satisfied either. 
Therefore, the random graph $G$ is not realizable by any $w$ under any map $f$. 
The randomness in this argument comes from the randomness in $G$. 
Hence the quantity $R(G,d)$ in Definition~\ref{prob:embed} would be zero. 
\end{proof}

Corollary~\ref{thm:lowerbound} uses a map $f$ that is not necessarily random. It also allows negative weights. Thus it disallows even a generalization of the context of Problem~\ref{prob:embed}.

By making a small modification in the counting argument mentioned above, it is possible to show that with high probability in sampling the tree, a random spanning tree of the complete graph on $n$ vertices cannot be realized if $d\leq n/2$.

\begin{restatable}[Realizing a random tree]{theorem}{thmLowerboundTree}
\label{thm:lowerbound-tree}%
Let $\mathcal{T}$ be the uniform distribution on the spanning trees of the complete graph $K_n$ with $n\geq 17$. Sample a tree $T=(V,E)\sim \mathcal{T}$. Let $d\leq n/2=O(n)$ be a positive integer. With probability at least $1-o(1)$ in the randomness of $T$, $T$ is not realizable under any map $V\Mapsto \{x,y\}^d$ for arbitrary $x,y\in \R{}$ and any weights $w_i\in \R{}, i\in [d]$. This means $R(T,d)=0$.
\end{restatable}
The proof is presented in Appendix~\ref{sec:lowerbound-proofs}.

\section{Conclusions}\label{sec:conclusions}
We defined a graph realization problem on random points and gave two algorithms, one for realizing graphs and the other for trees. 
We also proved that realizing random graphs requires a large target dimension.

\paragraph{Future work.}
Our realizing algorithms do not directly take advantage of any local or global structure of the tree/graph. 
The graph-realization algorithm samples from a family $\mathcal{A}$ of acyclic subgraphs; the ensemble of subgraphs has a bearing on the final bound. 
It is possible that we would get improved bounds if we focus on graphs with a certain combinatorial property, such as path graphs, planar graphs, etc. 
The effective resistance---a quantity intimately related to many algebraic properties of a graph---has appeared in our analysis. 
It would be interesting to see if one can design realization algorithms directly based on algebraic properties of the graph.

There could be graphs which need a higher target dimension than the $\Omega(n^2)$ bound from the random graphs. 
In general, it is an interesting prospect to reduce the necessity-sufficiency gap which currently stands at $\Omega(n\log n)$ vs. $\Omega(n)$ for trees and $\Omega(na^2\log n)$ vs. $\Omega(n^2)$ for graphs.

We have already seen that our algorithms work for weighted $\ell_1$ norms as well as weighted $\ell_2$ norms. 
Which other metric can we work with? Mahalanobis distance, perhaps, is a good candidate. 
An intriguing question is whether we can replace the ``uniformly random points'' in our problem with points generated from other stochastic processes. 
It is not obvious at this point how one can devise an algorithm for such a scenario. 
It is conceivable that the current analysis would work even if the map $f$ contains (sub-)Gaussian entries, but it still needs to be worked out.
At last but not the least, it is natural to ask how the bound on $d$ depends on the entropy of the data points.

\section{Acknowledgments}
We thank Ion Mandoiu for introducing to us the realization problem in the \cite{lp_tree} paper. 
We also thank Benjamin Fuller and Donald Sheehy for discussions and feedback which greatly improved the quality of the manuscript.
At last but not the least, Saad Quader would like to thank Tazrian Shinjon for her insightful comments.

\let \v \oldv

\bibliographystyle{alpha}
\bibliography{tree}

\appendix
\section{Omitted Proofs for Realizing Trees}\label{sec:tree-proofs}

\thmEdgeChernoff*
\begin{proof}
Let $X=D(u,v)$ and $Y=D(u^\prime, v^\prime)$. The random variables $X$ and $Y$ are sums of $d$ independent sub-Gaussian components, each component taking values in the interval $[0,s^2]$ of width $s^2$.

First, we want to show that $Y-X>0$ with high probability. Since Equation~(\ref{eq:gap-total}) tells us $\E (Y-X)=\sum{\delta_i}\geq d\delta$, it suffices to show that $(Y-\E Y) > \theta > (X+\E X)$ where $\theta=d\delta/2$.

Let $\mathcal{H}_{u,v}$ be the event that for an arbitrary edge $(u,v)\in E$, $X$ is ``too small'' compared to its expectation. Then, by Hoeffding's tail inequality, we have $
\Pr \mathcal{H}_{u,v} 
= \Pr \{\E X - X >\theta \} 
< \exp\left(-2\theta^2/\left(\sum_{i\leq d}{(s^2)^2}\right)\right)
= e^{-2\theta^2/s^4d}
= e^{-d\delta^2/2s^4}
$. Similarly, let $\mathcal{H}_{u,u^\prime}$ be the event that for an arbitrary non-edge $(u,u^\prime) \in \overline{E}$, $Y$ is ``too large'' compared to its expectation. In this case, we get $\Pr \mathcal{H}_{u,u^\prime} = \Pr \{Y - \E Y >\theta \} < e^{-d\delta^2/2s^4}$.

Now, a bad event $\mathcal{B}$ is one of the above two events for some vertex pair in $V\times V$. We want to show that the probability of this event is at most an inverse polynomial in $n$. Using a union bound over the $(n-1)$ tree edges and the remaining non-tree edges, we get
$
\Pr \mathcal{B}
\leq n^2 e^{-d\delta^2/2s^4} 
= \exp\left( 2\log n - d\delta^2/2s^4 \right)
$.

This probability will be at most $1/n = e^{-\log n}$ if $-2\log n + d\delta^2/2s^4 \geq \log n$, giving us \[d \geq \frac{6s^4\log n}{\delta^2}.\]

\end{proof}

\thmNonEdge*
\begin{proof}

Let $f:V\Mapsto \{x,y\}^d$. Consider the following process of generating the values $\{f(u)\}$: Select a set of $p_i\abs{E}$ edges uniformly at random out of the all possible $p_i\abs{E}$-element subsets of $E$. Set $\sigma(e)=1$ for these edges, and set $\sigma(e)=0$ for the remaining edges. Arbitrarily select a vertex $u$, and set $f(u)\in \{x,y\}^d$ uniformly at random. Set any unassigned vertex values as follows: for each edge $e=(a,b)\in T$, set $f(b)_i\gets f(a)_i$ if $\sigma(a,b)=1$, and set $f(b)_i\gets \overline{f(a)_i}$ otherwise where $\overline{x}=y$ and $\overline{y}=x$.

Notice that the distribution of $\{f(u)_i\}$ generated by the above process is identical to the observed distribution of $\{f(u)_i\}$. The good thing about this process is that the ``edge signs'' $\sigma$ have i.i.d. Bernoulli distribution with parameter $p_i$.

Let $P$ be the unique path from $u$ to $v$ along $T$, whose length is $t$. Let $S_t=\sum_{e\in P}{\sigma(e)}$. Define $c(t)\defeq \Pr[S_t\text{ is even }]$. Since $S_t$ has a binomial distribution with parameters $(t,p_i)$, it is not hard to show that $c(t)=\left(1+(2p_i-1)^t\right)/2$. Since $c(t)$ also equals $\PrAgree(u,v,i)$ conditioned on $d_T(u,v)=t$, the claim follows.
\end{proof}

\thmDeltai*
\begin{proof}
Fix coordinate $i$. Let $\epsilon\defeq \epsilon_i=p_i-p$ where $p_i$ is the fraction of agreeing edges at coordinate $i$. Substituting $p_i=p+\epsilon$ in Equation~(\ref{eq:deltai-simplified}) gives us $
\delta_i
\geq s^2q(2p+2\epsilon_i-1)(1-p-\epsilon_i)
=qs^2(2p-1)(1-p) + \lambda(p,\epsilon)$ where $
\lambda(p,\epsilon)
=qs^2\epsilon(3-4p-2\epsilon)$. It follows that $\delta_i$ is at least $qs^2(2p-1)(1-p)$ as long as $\lambda(p, \epsilon)\geq 0$. Since both $q$ and $s$ are strictly positive, this inequality gives us $
\epsilon 
\leq 1/2-2\alpha/\sqrt{n}
$ where we used $p=1/2+\alpha/\sqrt{n}$. This condition is equivalent to requiring $
p_i
\leq 1-\alpha/\sqrt{n}$. Recall that in Algorithm~\ref{alg:tree-second}, we have put a stronger requirement that $p_i$ must fall within the interval $[p,3/4]$ for $w_i$ to be $1$. Now we have to estimate $
q
=\Pr[w_i=1]
=\Pr[p\leq p_i \leq 3/4]$ which ensures $\delta_i\geq qs^2(2p-1)(1-p)$.

Let $Z$ be a random variable with a binomial distribution $B(n-1,1/2)$. 
Let $a \in [0,1]$, and define $\Tail(a) \defeq \Pr[Z \geq (n-1) a]$. 
According to Proposition~\ref{prop:anticonc-binomial}, $\Tail(p) > 1/2 - \alpha$. 
However, $\Tail(p) = q + \Tail(3/4)$, which implies $q \geq (1/2-\alpha) - \Tail(3/4)$. 

\begin{fact}
For any positive integer $n$ and $\beta \in [0,1/2]$ such that $\beta n$ is an integer,
\[
\sum_{k=0}^{\beta n}{\binom{n}{k} }
\leq 2^{H(\beta)n}
\, ,
\]
where $H(\beta)$ is the binary entropy function defined as $H(x)=-x\log_2x-(1-x)\log_2(1-x)$ for $x\in [0,1]$.
\end{fact}

Therefore, $
\Tail(3/4)
=\sum_{k=3n/4}^{n}{\binom{n}{k} }
=\sum_{k=0}^{n/4}{\binom{n}{k} }
\leq 2^{(H(1/4)-1)n}
\leq 2^{-0.18n}$ since $H(1/4)\leq 0.82$. Consequently, $
q
\geq (1/2-\alpha) - 2^{-0.18n }$. This value of $q$ is accompanied by $
\delta_i
\geq q s^2 (2p-1) (1-p)
=q s^2 (2\alpha/\sqrt{n}) (1-\alpha/\sqrt{n})$.
\end{proof}


\begin{restatable}[Anti-concentration]{proposition}{thmAnticonc}{
\label{prop:anticonc-binomial}
Let $n \geq 3$. 
Let $Z$ be a random variable with the binomial distribution $B(n-1,1/2)$. 
Suppose $\alpha \in (0, 1/2)$. 
Then 
\[
\left(\frac{1}{2}-\alpha \right)
< \Pr\left[ Z \geq \E Z + \alpha \sqrt{n} \right] 
< \frac{1}{2}
\, .
\]
}
\end{restatable}

\begin{proof}
It is easy to see that $\Pr[Z > \E Z + \alpha \sqrt{n}]$ is less than $1/2$ since the volume of a ``proper'' tail cannot exceed $1/2$.

Note that the peak of a binomial distribution remains relatively flat for small deviations from the mean.
The area under the pmf curve in that region can be closely overestimated by a (slightly larger) rectangle. 
This rectangle will have width $\alpha \sqrt{n}$ and height ${n \choose {n/2}}$ where $\sigma^2=n/4$ is the variance of a binomial distribution $B(n,1/2)$ and $\alpha$ is a small positive constant. 
We want to show that the mass in the tail beyond $n/2+\alpha \sqrt{n}$ is \emph{larger} than a constant. Let $N = n-1$.

\begin{align*}
q
&= \Pr[S_N \geq N/2 + \alpha \sqrt{n}] \\
&= 1/2 - \sum_{k=N/2}^{N/2 + \alpha\sqrt{n}}{ \frac{ {N \choose k } }{2^N} } \\
&> 1/2 -  (\alpha \sqrt{n}){N \choose {N/2} }2^{-N} \\
&\approx 1/2 - (\alpha \sqrt{n}) \left[ \frac{\sqrt{2}}{\sqrt{\pi}} \frac{2^N}{\sqrt{N} } \right] 2^{-N} \qquad \text{ (Stirling)} \\
&= 1/2 - \alpha  \sqrt{2/\pi} \sqrt{n/N} \\
&> 1/2 - \alpha \qquad \text{for } n \geq 3
\, .
\end{align*}
\end{proof}



\section{Omitted Proofs for Realizing Graphs}\label{sec:graph-proofs}
\claimGraph*

\begin{proof}
Let us use $\alpha = 1/4$ when invoking Algorithm~\ref{alg:tree-second}. 
In Algorithm~\ref{alg:tree-second}'s context, let $p=1/4+1/4\sqrt{n}$, and $q=\Pr[w_i=1]$. 
For every edge $e=(u,v)\in G$, let 
\[r(e) \defeq \Pr_{A}[e\in A],\] and $t$ be the length of the unique $u$-$v$ path in $A$ if it exists, and $\infty$ otherwise. Notice that $t=1$ if $(u,v)\in A$, and $t_e\geq 2$ otherwise. Define the quantity $\PrAgree(e,i)$ for an arbitrary edge $e=(u,v)\in G$ conditioned on the event that $w_i=1$.

\begin{align*}
\PrAgree(e,i\vert e\in G)
&= r(e) \PrAgree(e,i\vert e\in A) + (1-r(e)) \PrAgree(e,i\vert e\not\in A)\\
&= r(e) \left( \PrAgree(e,i\vert e\in A) - \PrAgree(e,i\vert e\not\in A) \right) + \PrAgree(e,i\vert e\not\in A)
\end{align*}
This implies,\begin{align*}
&\PrAgree(e,i\vert e\in G) - \PrAgree(e,i\vert e\not \in G)\\
&= r(e) \left( \PrAgree(e,i\vert e\in A) - \PrAgree(e,i\vert e\not\in A) \right) + \PrAgree(e,i\vert e\not\in A) - \PrAgree(e,i\vert e\not \in G)\\
&= r(e) \left( \PrAgree(e,i\vert e\in A) - \PrAgree(e,i\vert e\not\in A) \right)
\end{align*}
since conditioned on any $A$, the last two terms are the same. Continuing, the above quantity equals \begin{align*}
&= r(e) \left([1+(2p_i-1)]/2 - [1+(2p_i-1)^t]/2\right) \text{ using Claim~\ref{claim:inprod-tree-second-nonedge}}\\
&\geq r(e) \left((2p_i-1) - (2p_i-1)^2\right)/2 \text{ since }t\geq 2\\
&=r(e) \Omega(1/\sqrt{n}) \text{  by Claim~}\ref{claim:deltai}\\
&= \Omega(r/\sqrt{n}) \text{  where } r=\min_{e}{r(e)}.
\end{align*}

This implies $\delta_i=\Omega(r/\sqrt{n})$ from Definition~\ref{def:deltai} when $s$ is $O(1)$ and $n\geq 20$. By an application of Lemma~\ref{lemma:edge-chernoff}, we see that $d=\Omega(\dfrac{n}{r^2}\log n)$ is sufficient to realize $G$.
\end{proof}

\section{Omitted Proofs for Random Graphs}\label{sec:lowerbound-proofs}
A well-known theorem in convex geometry is Radon's theorem, which relates the linear separability of point-sets with the ambient dimension. It states that it is always possible to label any collection of at least $d+2$ points in $\R{d}$ into two subsets which are not linearly separable.

\begin{theorem}[Radon's Theorem]\label{thm:radon}
If $B$ is a set of $M$ points in $\R{d}$ with $M\geq d+2$, there exists a partition $B=E\djunion \overline{E}$ such that the convex hulls of $E$ and $\overline{E}$ has nonempty intersection. Consequently, there can be no hyperplane separating $E$ from $\overline{E}$.
\end{theorem}

In our context, Theorem~\ref{thm:radon} says ``for every map $B=g(V\times V)$ there exists a nonseparable partition of $B$''. However, we want to show that ``there exists a graph $G=(V,E)$ such that the two subsets $g(E)$ and $g(\overline(E))$ of $B$ are nonseparable for every $f$.'' This requires a change in the order of the quantifiers (the ``for every'' and ``there exists'') in the statement of Radon's theorem. Fortunately, it turns out that a random partition just works: it effectively lets us exchange the said quantifiers. Moreover, a uniformly random partition of $V\times V$  means $G$ is an \erdosrenyi random graph $G\sim \mathcal{G}(n,1/2)$. This notion is captured in Theorem~\ref{thm:nonseparable}, which is ``expensive'' than Radon's theorem:  the number of points needs to be at least $6d$ (roughly speaking) instead of just $d+2$. Additionally, and the claim holds true with high probability. 

\thmNonSeparable*

The main goal of this section is to present a proof. We prepare by developing two propositions.


\begin{restatable}[Nonseparability via $d$-supported hyperplanes]{proposition}{thmPropNoPressing}
\label{prop:no-pressing}
Let $n\geq 7$ and $M=\binom{n}{2}$ be two integers. Let $B$ be an arbitrary set of $M$ points in $\R{d}$ where $d \leq M/ 6$. Let $B=E \djunion \overline{E}$ be a uniformly random partition of $B$. Then, with probability $1-o(1)$, the convex hulls of $E$ and $\overline{E}$ cannot be separated by a hyperplane supported on any $d$ points of $B$.
\end{restatable}
\begin{proof}

Let $H_B$ be the set of hyperplanes that pass through exactly $d$ points of $B$. This implies $\abs{H_B}=\binom{M}{d}$. Now consider a hyperplane $h\in H_B$ which separates the bipartition $E\djunion \overline{E}=B$. Fix $h$. We claim that the number of distinct binary labelings $B\Mapsto \{E, \overline{E}\}$ that $h$ can separate is $2\cdot 2^d=2^{d+1}$, as follows: two choices for the symmetry of $E$ and $\overline{E}$ with respect to $h$ (one gets the label $E$ and the other gets $\overline{E}$), and $2^d$ choices for the classification of the $d$ points supporting the hyperplane into $\{E,\overline{E}\}$. By a union bound over all hyperplanes, the number of distinct decorations of the points of $B$ that can be separated by \emph{some} hyperplane is at most $2^{d+1}\abs{H_B} = 2^{d+1}\binom{M}{d}$.

However, the total number of labelings $B\Mapsto \{E, \overline{E}\}$ is $2^M$. Let $p(M,d)$ be the probability that the two convex sets induced by a \emph{random} labeling $r^*$ is separated by \emph{some} hyperplane $h\in H_B$. That is,
\begin{align*}
p(M,d)&= \Pr_{\substack{r\sim \{0,1\}^M\\h\sim H_B}}[h \text{ separates } E \text{ from } \overline{E}]\\
&\leq \frac{2^{d+1}\binom{M}{d}}{2^M} 
\leq \frac{\left(\frac{Me}{d}\right)^d}{2^{M-d-1}} 
= \frac{(\alpha e)^d}{2^{d(\alpha-1)-1}}
\end{align*}
where $M=\binom{n}{2}=\alpha d$ for some $\alpha > 1$. This quantity will be at most $1/n$ for all $n\geq 7$ if we set $\alpha\geq 6$.

Therefore, the probability that no $d$-supported hyperplane $h\in H$ separates the random partition $E \cap \overline{E}$ is at least $1-1/n=1-o(1)$ when $n\geq 7$ and $d=O(n^2)$.

\end{proof}


\begin{restatable}[Separating hyperplanes imply pressing hyperplanes]{proposition}{thmPropSepImplyPress}
\label{prop:sep-imply-press}
Let $S$ be the affine subspace spanned by the points $B$. Let $d=\dim(S) \geq 2$. Let $B_0\djunion B_1=B$ be a partition of $B$ such that the convex hulls of $B_0$ and $B_1$ do not intersect. Then, there exists a hyperplane $h$  which separates $B_0$ and $B_1$ and moreover, it is supported on exactly $d$ points of $B$.
\end{restatable}
\begin{proof}

Let $B\subset \R{d}$. Let $\inprod{x}{y}=y^Tx=\sum_{i}{x_iy_i}$ for every $x,y\in \R{d}$.

Since the convex hulls of $B_0$ and $B_1$ do not intersect, the separating hyperplane theorem implies that there exists a hyperplane $h$ such that 
\begin{align*}
\inprod{b}{h} &\geq 1 \text{ for all } b\in B_0\\
\inprod{c}{h} &\leq 1 \text{ for all } c\in B_1
\end{align*} Let $\mathcal{L}$ be the above feasible linear system. We make the following claim.
\begin{claim*}
The feasibility polytope $P$ of the above linear system does not contain an affine linear subspace of dimension $1$. 
\end{claim*}
If the claim is true, $P$ will have a vertex $h^*$ that meets $d$ constraints, each a $d-1$ dimensional facet of $P$. This vertex does in fact corresponds to a separating hyperplane that satisfies $d$ linear constraints of $\mathcal{L}$ with equality. Since each constraint is given by one point of $B$, this implies $h^*$ is supported by $d$ points in $B$.

It remains to prove the claim. For the sake of contradiction, assume that $P$ contains an affine subspace $H=\{h\}$ of dimension $1$ defined by the equation $h=h_x+\lambda h_y$ for some $h_x, h_y\in P$ and all $\lambda \in \R{}$. 

Suppose there exists a point $b\in B$ that is not orthogonal to the (separating) hyperplane $h_y$ i.e., $\inprod{b}{h_y}\neq 0$. Such a point $b$ will always exist because otherwise, all points of $B$ would lie on the same line (normal to $h_y$) and $d$ would be one, violating the condition that $d=\dim(\text{span}(B))\geq 2$. Without loss of generality, assume that $b \in B_0$.

Since $h\in H\subseteq P$, it implies that for all $\lambda\in \R{}$,
\begin{align*}
\inprod{b}{x+\lambda y} &\geq 1 \text{ for all } b\in B_0\\
\inprod{c}{x+\lambda y} &\leq 1 \text{ for all } c\in B_1
\end{align*}
Thus we can freely choose $\lambda_1, \lambda_2\in \R{}$ and write $h_1=h_x+\lambda_1 h_y, h_2=h_x+\lambda_2 h_y$ such that $h_1, h_2\in H\subseteq P$ and $\inprod{b}{h_1} \leq 1 \leq \inprod{b}{h_2}$. Intuitively speaking, we have translated a separating hyperplane $h_1$ to a new separating hyperplane $h_2$ along the direction $h_y$. However, there is now a point $b\in B$ which ``satisfies'' only one of the hyperplanes $h_1,h_2$ but not both. This is a contradiction, since both $h_1,h_2 \in H \subseteq P$ are two feasible solutions of $\mathcal{L}$. Therefore, the claim must be true.

\end{proof}

\paragraph{Proof of Theorem~\ref{thm:nonseparable}.} 
\begin{proof}{\textbf{(of Theorem~\ref{thm:nonseparable})}}
Apply Proposition~\ref{prop:no-pressing} along with the contrapositive of Proposition~\ref{prop:sep-imply-press}. Together, they imply that if $d$ is at most $M/6$, a uniformly random partition $E\djunion \bar{E}=B\subset \{x,y\}^d$ is nonseparable with probability $1-o(1)$ in the randomness of the partition.
\end{proof}

While Theorem~\ref{thm:nonseparable} applies to random graphs, it is possible to modify Proposition~\ref{prop:no-pressing} so that a similar statement applies to random graphs with $(n-1)$ edges.

\thmLowerboundTree*

\begin{proof}
One can make an argument similar to that in the proof of Corollary~\ref{thm:lowerbound}. The only place to change would be the proof of Proposition~\ref{prop:no-pressing}. Let $n_t$ be the number of colorings (trees) that are separable (realizable) by some $d$-supported hyperplane.  Although we do not know an exact estimate on $n_t$, it is certainly smaller than the number of all colorings separable by some hyperplane passing through $d$ points. From the proof of Proposition~\ref{prop:no-pressing}, we know that this number is $\binom{M}{d}2^{d+1}$. Hence
\[
n_t
\leq \binom{M}{d}2^{d+1}
\leq \left(\frac{Me}{d}\right)^d 2^{d+1}
\leq \left(\frac{(n^2/2)e}{d}\right)^d 2^{d+1}
= n^{2d}(e/2d)^d 2^{d+1}
\]
since $M=\binom{n}{2}\leq n^2/2$. By Cayley's formula, the number of labeled trees on $n$ vertices is $n^{n-2}$. Thus the probability $p$ that a coloring, chosen uniformly at random from the colorings corresponding to random spanning trees, is  
\[
p=\frac{n_t}{n^{n-2}} \leq \frac{(e/2d)^d 2^{d+1}}{n^{n-2-2d}}.\]
$p$ will be less than $1/n$ if 
\begin{align*}
&\frac{(e/2d)^d 2^{d+1}}{n^{n-2-2d}}<\frac{1}{n}\\
\implies & (e/2d)^d 2^{d+1}<n^{n-3-2d}\\
\implies & d\log (e/2d)+(d+1)\log 2 < (n-3-2d)\log n\\
\implies & d - d\log(2d)+(d+1)\log 2 < (n-3-2d)\log n \\
\implies & d +\log 2 + (2d+3-n)\log n < d\log(d).
\end{align*}	
By setting $d\leq n/2$, the left hand side is at most $n/2 +\log 2 + 3\log n$, which is strictly smaller than the right hand side $(n/2)\log (n/2)$ when $n\geq 17$.

Therefore, with probability $1-o(1)$ there exist a random tree on $n$ vertices which is not realizable by any real weights and any mapping $V\Mapsto \{x,y\}^d$ when $d\leq n/2$ and $n\geq 17$.

\end{proof}

\end{document}